\documentclass[10pt,twocolumn,amsmath,amssymb,aps,pra,nofootinbib]{revtex4-2}

\usepackage{graphicx} 
\usepackage{lipsum}   

\pdfoutput=1  
\usepackage[utf8]{inputenc}
\usepackage{amsmath,amsfonts,amssymb}

\usepackage{nameref}
\usepackage{color}
\usepackage{graphicx}
\usepackage{comment}
\usepackage{color}
\usepackage{physics}
\usepackage{mathtools}
\usepackage{mathrsfs}
\usepackage{booktabs}
\usepackage{bm}
\usepackage{amsthm}
\usepackage[english]{babel}
\newtheorem{theorem}{Theorem}

\usepackage{amsmath}
\usepackage{physics}
\usepackage{bm}
\usepackage[page,toc,titletoc,title]{appendix}
\usepackage{fixmath}
\usepackage{dsfont}
\usepackage{hyperref}
\hypersetup{
    colorlinks,
    citecolor=blue,
    filecolor=black,
    linkcolor=blue,
    urlcolor=blue
}
\usepackage{graphicx}
\usepackage{graphicx,accents}

\usepackage{mathtools}
\DeclarePairedDelimiter{\ceil}{\lceil}{\rceil}

\begin{document}
\title{Quantum generation of stochastic processes: spectral invariants and memory bounds} 

\author{Magdalini Zonnios}
\email{zonniosm@tcd.ie (corresponding author)}
\affiliation{School of Physics, Trinity College Dublin, College Green, Dublin 2, D02 K8N4, Ireland}
\affiliation{Trinity Quantum Alliance, Unit 16, Trinity Technology and Enterprise Centre, Pearse Street, Dublin 2, D02 YN67, Ireland}

\author{Alec Boyd}
\email{alecboy@gmail.com}
\affiliation{Beyond Institute for Theoretical Science, San Francisco, CA, USA}
\affiliation{School of Physics, Trinity College Dublin, College Green, Dublin 2, D02 K8N4, Ireland}
\affiliation{Trinity Quantum Alliance, Unit 16, Trinity Technology and Enterprise Centre, Pearse Street, Dublin 2, D02 YN67, Ireland}

\author{Felix C. Binder}
\email{felix.binder@tcd.ie}
\affiliation{School of Physics, Trinity College Dublin, College Green, Dublin 2, D02 K8N4, Ireland}
\affiliation{Trinity Quantum Alliance, Unit 16, Trinity Technology and Enterprise Centre, Pearse Street, Dublin 2, D02 YN67, Ireland}

\date{\today}

\begin{abstract}
Stochastic processes abound in nature and accurately modeling them is essential across the quantitative sciences. They can be described by hidden Markov models or by their quantum extensions. These models explain and give rise to process outputs in terms of an observed system interacting with an unobserved memory. Although there are infinitely many models that can generate a given process, they can vary greatly in their memory requirements. It is therefore of great fundamental and practical importance to identify memory-minimal models. This task is known to be hard in general, due to both the number of generating models, and the lack of invariant features that determine elements of the set. In general, it is forbiddingly difficult to ascertain that a given model is minimal. Addressing this challenge, we here identify spectral invariants of a process that can be calculated from any model that generates it. This allows us to determine strict bounds on the quantum generative complexity of the process—its minimal memory requirement. We then show that the bound is raised quadratically when we restrict to classical operations. This is an entirely quantum-coherent effect, as we express precisely, using the resource theory of coherence. Finally, we demonstrate that the classical bound can be violated by quantum models.  
\end{abstract}
\maketitle
\section{Introduction}Correlated stochastic processes are abundant in nature and can be understood in terms of interactions between an observed system and a hidden unobserved memory.  The memory stores information, carrying it between past and future to create temporal correlation in the observed data. There are many established ways to model processes with memory, including hidden Markov models~\cite{rabiner_tutorial_1989,upper_theory_1997}, finite state automata~\cite{rabin_probabilistic_1963,Gruska_potential_2015,Tian_experimental_2019}, and more recently so-called quantum hidden Markov models (QHMMs)~\cite{giovannetti_quantum_2008,wiesner_computation_2008,oneill_hidden_2012,Clark_hidden_2015,monras_quantum_2016,Cholewa_QHMM_2017_2,SrinivasanLearning_2018,Aloisio_PRXQuantum.4.020310}, which show advantages over their classical counterparts~\cite{Gu_Quantum_2012,mahoney_occams_2016,riechers_minimized_2016,aghamohammadi_extreme_2017,lund_quantum_2017,Garner_adv_2017,elliott_superior_2018,Schuld_learning_2018,RareEventAghamohammadi2018,loomis_strong_2019,ghafari_dimensional_2019,elliott_memory-efficient_2019, ho_robust_2020,elliott_Extreme_2020,blank_quantum-enhanced_2021,korzekwa_quantum_2021,elliott_quantum_2022,wu_implementing_2023,hangleiter_computational_2023,banchi_accuracy_2024,elliott_embedding_2024}. The memory requirements of these models vary with respect to the complexity of the underlying process. In general, for a given process, there are multiple consistent ways to accurately model the data, each of which constitutes a \emph{presentation} \cite{ellison2011informationChaos}.  The presentation and corresponding memory size of a process, however, are not unique: many distinct presentations can give rise to identical behaviors on the observable system~\cite{shalizi_computational_2001,ay_reductions_2005,monras_hidden_2012,riechers_minimized_2016,ruebeck_prediction_2018,jurgens_divergent_2021}. 

Minimal generative models are the privileged presentations that require the smallest memory. These models serve as a benchmark for understanding the generative complexity of the output process as they are irreducible (or incompressible) in the sense that they require the least amount of computational memory resources. The pursuit of minimal models is both of fundamental and practical interest~\cite{mahoney_occams_2016,wu_implementing_2023,yang_2024_dimensionreductionquantumsampling}, yet the challenge of identifying them or generically determining whether or not a model is minimal remains a difficult and open problem~\cite{lohr_predictive_2012,ruebeck_prediction_2018}. In fact, it has been suggested that identifying minimal generative models is generically NP-hard~\cite{Paulson2017minimumprobabilisticfinitestate}.
This difficulty stems from two key issues: Firstly, the full set of QHMMs over which to minimize is generally infinite. And secondly, the absence of known invariant features of the presentations that define equivalent process makes it hard to even specify the set of candidate QHMMs for optimization. Here we take a step towards resolving this challenge by identifying a spectral invariant property of the process which we can calculate from any any valid presentation of the process (whether minimal or not). This allows us to place strict bounds on the necessary memory to generate a stochastic process via repeated quantum operations.

We achieve this by first defining a QHMM in terms of the repeated application of a single quantum instrument which acts sequentially on memory over time in order to generate a process. Then, building on existing works which frame probability distributions in the language of tensor networks~\cite{stokes_probabilistic_2019,Glasser_Expressive_2019,li_connecting_2024,glasser_probabilistic_2020,Khavari_lower_2021,wall_tree-tensor-network_2021,lu_tensor_2021, adhikary_quantum_2021,christandl_resource_2023,harvey_sequence_2023,rieser_tensor_2023,Liao_Decohering_2023,li_2023_efficientquantummixedstatetomography,li_connecting_2024} we express a QHMM as a tensor chain, as has also been done for classical HMMs. Unlike classical HMMs however, QHMMs can exploit coherences to encode imperfect statistical distinguishability into imperfect state distinguishability via non-orthogonal memory states. This in turn introduces the possibility of reductions to the minimal memory required for complete sampling.  \\\\

It is known that classical generative HMMs can exhibit both entropic and topological advantages over minimal predictive models which are additionally restricted to be unifilar~\cite{shalizi_computational_2001, lohr_predictive_2012, ruebeck_prediction_2018}. Generative models aim to reproduce the full joint statistics of a process, in contrast to predictive models, which are designed to infer future outputs from past observations. While quantum advantage and minimality are relatively well understood in the predictive setting~\cite{binder_practical_2018,Liu_optimal_2019,elliott_memory-efficient_2019,elliott_Extreme_2020}, and entropic advantages have been demonstrated for quantum generative models~\cite{elliott_memory_2021}, less is known about the corresponding topological advantages in the generative case. In particular, there is currently no general framework for determining whether, and under what conditions, such an advantage persists in terms of the number of internal states. Here, we provide a framework for bounding the number of internal states required by classical and quantum generative models, and demonstrate through explicit construction that a topological quantum advantage can arise in this setting. Our results offer new structural insights into the origin of this advantage and clarify the role of coherence as a resource in generative modeling.

The structure of this Letter is as follows. We define a quantum hidden Markov model (QHMM) in terms of a quantum instrument $\mathcal{I}$ together with an initial state $\rho$. We then introduce the generative topological complexity, $c_Q(\overrightarrow{X})$ of a process $\overrightarrow{X}$ as the minimal memory required by any QHMM to generate a process. By expressing QHMMs as tensor chains, we show that the distinct, non-zero spectrum of their transfer operators serve as a process invariant, allowing us to derive a lower bound on $c_Q(\overrightarrow{X})$ which we present in Thm.~\eqref{theorem:quantum tc bound}. Next, we show in Thm.~\eqref{theorem:classical tc bound} that restricting to classical (i.e., Strictly Incoherent) operations quadratically raises the bound, highlighting that reductions in topological complexity beyond minimal classical HMMs require quantum coherence. Finally, we present an example of a process to prove the existence of quantum models with memory size below the classical bound.

\section{Quantum simulation of stochastic processes via QHMMs}
\begin{figure}
    \centering
\includegraphics[scale=0.8]{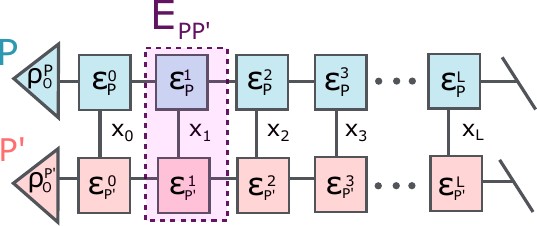}
    \caption{The transfer operator $\mathbb{E}_{PP'}=\sum_x\mathcal{E}_{P}\otimes\mathcal{E}_{P'}$ (highlighted in the purple box) is obtained by contracting the output of two tensor chains for processes $P$ and $P'$ which are obtained by a sequence of quantum instruments acting on memory.}
    \label{fig:TO}
\end{figure}
A stochastic process $\overrightarrow{X}$ is a sequence of random variables $\overrightarrow{X}:=X_0X_1\dots$. The outputs $x_i$ at each time $t_i$ are distributed according to $X_i$ and drawn from an output alphabet $\mathcal A$. Any such process can be simulated by the repeated application of a quantum instrument $\mathcal{I}$ acting on quantum memory $\rho_m$ as $\mathcal{I}[\rho_m]:=\sum_x\mathcal{E}^x[\rho_m]\otimes\dyad{x}$ such that, upon observing the outcome $x$, $\mathcal{I}$ updates a classical register to $\dyad{x}$ and the internal memory state $\rho_m$ via the completely positive (CP) map $\mathcal{E}^x$~\cite{vieira_temporal_2022}. The instrument $\mathcal E \equiv \{\mathcal{E}^x\}$ also induces a positive operator-valued measure (POVM) $\{M_x\}$ on the memory space with elements $M_x=\sum_\alpha K_\alpha^{(x)\dag} K_\alpha^{(x)}$ specified for any choice of Kraus operators $K_\alpha$ of the corresponding map $\mathcal{E}^x$, $\mathcal{E}^x[\rho]=\sum_\alpha K_{\alpha}^{(x)}\rho K_\alpha^{(x)\dag}$. $\mathcal{E}=\sum_x{\mathcal{E}^x}$ is completely positive and trace preserving (i.e., $\trace(\mathcal{E}[\rho])=\trace\left(\sum_x{\mathcal{E}^x}[\rho]\right)=\trace(\rho)$). The probability of obtaining outcome $x$ when the system is in the state $\rho$ follows the Born rule, $P(x)_\rho=\trace(M_{x}\rho)$ so that given an initial memory state $\rho_0$, the probability of measuring the sequence $x_{0:L}$ is
\begin{align}\label{eq:ProbAsSequenceOfInstruments}
    \Pr(x_{0:L})_{\rho_0} = \tr\left(\mathcal{E}^{x_L}\circ\dots\circ\mathcal{E}^{x_1}\circ\mathcal{E}^{x_0}[{\rho_0}]\right). 
\end{align}
We define a QHMM $R$ as the tuple of a quantum instrument $\mathcal{I}_R$ and an initial state $\rho_R$: $R=(\mathcal{I}_R,\rho_R)$ and denote the generated probability distribution as $R(x_{0:L})$. 
Thus, each model $R$ specifies a stochastic process. Conversely, a given stochastic process can be produced by many different QHMMs, as any QHMM $R$ which outputs sequences of $\overrightarrow{X}$ with the correct probabilities via Eq.~\eqref{eq:ProbAsSequenceOfInstruments}, is a generative model of $\overrightarrow{X}$. We denote the set of generating models for $\overrightarrow{X}$ as $\mathbf{R}_{\overrightarrow{X}}$.

\section{Topological generative complexity of QHMMs} 
The variety in valid QHMMs for a given process implies that the input space of the instruments may vary from model to model, with some requiring more, and others less memory in order to generate the same process faithfully.  Within $\mathbf{R}_{\overrightarrow{X}}$, we identify those that require the smallest amount of memory as topologically-minimal. The minimal memory required to generate a process -- its topological generative complexity -- has been defined and studied in the classical context~\cite{upper_theory_1997,lohr_predictive_2012,ruebeck_prediction_2018}. In analogy, in the quantum domain the topological generative complexity $c_Q$ of $\overrightarrow{X}$ is
\begin{align}\label{eq:TCgQDef}
        c_Q({\overrightarrow{X}}):=\min_{R \in \mathbf{R}_{\overrightarrow{X}}}\{\log[{\dim\left[\mathcal{H}_{R}^{in}\right]}]\},
\end{align}
where $\mathcal{H}_{R}^{in}$ is the Hilbert space upon which the instruments $\mathcal{I}^{R}$'s input density operators are defined. In other words, the memory size of a model is the base-2 logarithm of the Hilbert space dimension required to encode its internal memory states and the measure $c_Q({\overrightarrow{X}})$ corresponds to the minimal memory size among all generating models that gives rise to $\overrightarrow{X}$. In general, the minimization in Eq.~\eqref{eq:TCgQDef} is highly non-trivial and the set of possible QHMMs for a given process is infinite. However, as we will show, it is possible to identify invariant features of the set, and thus place strict bounds on $c_Q(\overrightarrow{X})$. 

We present our main by introducing the transfer operator $\mathbb{E}$ (as shown in Fig.~\ref{fig:TO}) and its vectorised form $\mathbf{E}$. For any process $\overrightarrow{X}$, all models ${R}\in\mathbf{R}_{\overrightarrow{X}}$ which generate the process have transfer operators $\mathbf{E}_{RR}$ with an equivalent distinct spectrum. This key property significantly constrains the possible generating models, whether classical or quantum. By identifying this invariant feature, we directly establish place bounds on the required generating memory, as it depends on the size of the invariant spectrum. 

Before formally introducing the transfer operator let us note that what we mean when we say that two models $R$ and $Q$ generate the same process is that $R(x_{0:L})= Q(x_{0:L})~~\forall~x_{0:L},~L$. We note that if $R_L=Q_L~\forall~L$, then
\begin{align}\label{eq:innerProductRQ}
     \sum_{x_{0:L}}R(x_{0:L})^2 = \sum_{x_{0:L}}Q(x_{0:L})^2.
\end{align}
Then we introduce the transfer operator
\begin{align}\label{eq:transferOperator}
    \mathbb{E}_{RQ} := \sum_x\mathcal{E}_{R}^x\otimes\mathcal{E}_{Q}^{xT},
\end{align}
which acts on product states $\rho_R\otimes\rho_Q$. The transfer operator captures how correlations between two quantum instruments propagate in time.  If $\mathcal{E}_R^{(x)}$ and $\mathcal{E}_Q^{(x)}$ have Kraus representations $\{K_{\alpha_R}^{(x)}\}$ and $\{K_{\alpha_Q}^{(x)}\}$, then we can express $\mathbb{E}_{RQ}$ in the vectorised picture as~\cite{milz_introduction_2017}
\begin{align}
        \mathbb{E}_{RQ}[\rho_{RQ}]&\to {\mathbf{E}_{RQ}\ket{\rho_{RQ}}\!\rangle}\\
    &={\sum_{\substack{\alpha_{R}\\\alpha_{Q}\\x}}K_{\alpha_R}^{(x)}\otimes K_{\alpha_R}^{(x)*}\otimes K_{\alpha_Q}^{(x)}\otimes K_{\alpha_Q}^{(x)*}\ket{\rho_{RQ}}\!\rangle.}
\end{align}
{where $\mathbf{E}_{RQ}$ is the vectorised (A form) version of the transfer operator $\mathbb{E}_{RQ}$, and $|\rho_{RQ}\rangle\!\rangle=|\rho_{R}\rangle\!\rangle\otimes |\rho_{Q}\rangle\!\rangle$ is the vectorised tensor product of initial states of the respective QHMMs.}
{Together with Eq.~\eqref{eq:ProbAsSequenceOfInstruments}, we have,}
\begin{align}\label{eq:VecEquivalence}
        {\sum_{x_{0:L}}R(x_{0:L})Q(x_{0:L})} &=\trace (\mathbb{E}_{RQ}^{\circ L}\left[\rho_{R}\otimes\rho_{Q}\right]) \\&=     \langle\langle \mathds{1} | \mathbf{E}_{RQ}^L|\rho_{RQ}\rangle\rangle,
\end{align}
where $|\mathds{1}\rangle\rangle$ is the vectorised identity operator {acting as the trace on the output spaces of $\mathbf{E}_{RQ}$}. Then we can we re-express Eq.~\eqref{eq:innerProductRQ} as
\begin{align}\label{eq:VecEquivalenceVectorised}
    \langle\langle \mathds{1} | \mathbf{E}_{RR}^L|\rho_{RR}\rangle\rangle= \langle\langle \mathds{1} | \mathbf{E}_{QQ}^L|\rho_{QQ}\rangle\rangle.
\end{align}
By the condition in Eq.~\eqref{eq:innerProductRQ}, if two QHMMs $R$ and $Q$ generate the same process then Eq.~\eqref{eq:VecEquivalenceVectorised}
must hold for all~$L$.
\section{Bounds on memory dimension of generative models}
In order to arrive at the central result of this letter we identify two invariant spectral features of the transfer operator of any QHMM which generate the same process. First, we introduce the set of distinct non-zero eigenvalues in the spectrum of $\mathbf{E}_{RR}$, $\Lambda_R=\{ \lambda_{R}~|~\exists~\Pi_{\lambda_R} 
\neq 0 : \mathbf{E}_{RR}\Pi_{\lambda_R} = \lambda \Pi_{\lambda_R} \}$  (and likewise for $\Lambda_Q$ and $\mathbf{E}_{QQ}$).  
We also define corresponding coefficients  $\alpha_{\lambda}=\langle\langle\mathds{1} |\Pi_{\lambda}|\rho\rangle\rangle$, which relate eigenprojectors $\Pi_{\lambda}$ to their respective eigenvalues $\lambda: \mathbf{E}_{RR}=\sum_{\lambda \in \Lambda_{R}} \lambda \Pi_\lambda$.  This spectral decomposition exists as long as $\mathbf{E}_{RR}$ is diagonalisable \footnote{{While we would expect two full-rank instruments to give rise to a diagonalisable transfer operator, we note that we can revert to the Jordan Normal Form in the non-diagonalisable case. See Appendix.~\ref{sec:ProofofLemma1} for details.}}.
\begin{theorem}\label{lemma:EquivalenceOfPQ}
If two quantum hidden Markov models $R=(\mathcal{I}^R,\rho_R)$ and $Q=(\mathcal{I}^Q,\rho_Q)$ generate the same stochastic process it necessarily holds that 
 \begin{align}\label{eq: lem 1,Eq1}
     \Lambda_{\overrightarrow{X}}:=\Lambda_R &= \Lambda_Q,
\end{align}
and that
\begin{align}\label{eq: lem 1,Eq2}
    \lambda_R=\lambda_Q \Rightarrow \alpha_{\lambda_R} = \alpha_{\lambda_Q}.
\end{align}
\end{theorem}
\begin{proof}
Expressing the vectorised transfer operators in their eigen-decompositions, $\mathbf{E}=\sum_\lambda \lambda \Pi_\lambda$, we can re-write Eqn.~\eqref{eq:VecEquivalenceVectorised} as
\begin{align}\label{eq:Lem1 proof eq1}
    \sum_{\lambda_R \in \Lambda_R} \lambda_R^L \alpha_{\lambda_R} =     \sum_{\lambda_Q \in \Lambda_Q} \lambda_Q^L \alpha_{\lambda_Q}, 
\end{align}
If we define $\Lambda_{RQ}:=\Lambda_R\cup \Lambda_Q$ then we may subtract the left from the right side of the expression to obtain
\begin{align}\label{eq:Lem1 eq1}
\sum_{\lambda_{RQ}\in\Lambda_{RQ}}\lambda_{RQ}^L\alpha_{\lambda_{RQ}}  = 0,
\end{align}
where $\alpha_{\lambda_{RQ}} := \alpha_{\lambda_R=\lambda_{RQ}} - \alpha_{\lambda_Q=\lambda_{RQ}}$, is the difference between the projectors corresponding to the eigenvalue $\lambda_{RQ}$. The solution to Eq.~\eqref{eq:Lem1 eq1} is given by considering the Vandermonde matrix in this expression to show that we must have $\alpha_{\lambda_{RQ}}=0$ for all $\lambda_{RQ}$ and consequently {$\alpha_{\lambda_R}=\alpha_{\lambda_Q}$ for all $\lambda_R$, $\lambda_Q$ for all $\lambda_R=\lambda_Q$}. Additional details can be found in Appendix~\ref{sec:ProofofLemma1}.
\end{proof}
Theorem~\eqref{lemma:EquivalenceOfPQ} shows that the distinct, non-zero spectrum is an invariant of the process ($\Lambda_{\overrightarrow{X}}$), irrespective of the QHMM which generates it. The freedom among generating models thus lies in the second part of Thm.~\eqref{lemma:EquivalenceOfPQ} which shows that the eigenprojectors corresponding to a given eigenvalue need not be equal for all models. Moreover since the set $\Lambda_{\overrightarrow{X}}$ is agnostic to the presence of degenerate eigenvalues, alternative models may still include multiple different eigenprojectors that correspond to the same eigenvalue. This highlights that the only possible operations that can be made to a QHMM and still have it generate the same process, are those that preserve the spectrum of the corresponding transfer operator. We note that the converse does not necessarily hold: there may exist alternative processes which have transfer operators with the same distinct spectrum, i.e., $\Lambda_{\overrightarrow{X}}=\Lambda_{\overrightarrow{X}'}$ for $\overrightarrow{X}\neq{\overrightarrow{X}'}$. {This is due to the fact that while $\overrightarrow{X}\neq{\overrightarrow{X}'}$ implies the equivalence presented in Eq.~\eqref{eq:innerProductRQ}, the equation does not conversely imply $R(x_{0:L})=Q(x_{0:L})$ for all $x_{0:L}$ and all $L$.}
\\\\
Nevertheless, by fixing the number of unique eigenvalues required for generating a given process, we are able to set strict lower bounds on the dimensions of the required memory space. This leads to the first main result of our paper presented in Thm.~\eqref{theorem:quantum tc bound}.

\begin{theorem}\label{theorem:quantum tc bound}
        The generative topological complexity $c_Q(\overrightarrow{X})$ is bounded from below by the ceiling of the fourth root of the number of unique eigenvalues in the spectrum of any valid transfer operator of the process, i.e., 
        \begin{align}\label{eq:cQInequality}
            c_Q(\overrightarrow{X}) \geq \log{\ceil{{|\Lambda_{\overrightarrow{X}}|}^\frac{1}{4}}}
        \end{align}
where $\Lambda_{\overrightarrow{X}}$ can be determined via any QHMM that generates $\overrightarrow{X}$.
\end{theorem}
\begin{proof}    
The QHMM $R_{m}$ which minimizes the right hand side of Eq.~\eqref{eq:TCgQDef} for a process $\overrightarrow{X}$ will have inputs defined on $\mathcal{L}(\mathcal{H}_{R_{m}}^{in})$ such that the corresponding vectorised transfer operator $\mathbf{E}_{R_{m}R_{m}}$ will have inputs on $\mathcal{H}_{R_{m}}^{in \otimes 4}$ with dimension $\dim(\mathcal{H}_{R_{m}}^{in \otimes 4}) = \dim(\mathcal{H}_{R_{m}}^{in})^4$. Then, since the topological complexity is $c_Q(\overrightarrow{X})=\log(\dim(\mathcal{H}_{R_{m}}^{in}))$ and by Theorem.~\eqref{lemma:EquivalenceOfPQ} the transfer operator for all valid generators of the process will have the same (distinct) spectrum $\Lambda_{\overrightarrow{X}}$, we know that $\dim(\mathcal{H}_{R_{m}}^{in})^4\geq |\Lambda_{\overrightarrow{X}}|$. Taking the fourth root of both sides and the ceiling (to ensure we achieve the next largest integer value), we arrive at Eq.~\eqref{eq:cQInequality}.  
\end{proof}
We have shown that the number of internal states required by any model which generates the stochastic process $\overrightarrow{X}$ is lower bounded by the spectrum $\Lambda_{\overrightarrow{X}}$ given by the transfer operator of any generator of $\overrightarrow{X}$.

\section{Bounds on stochastic processes generated via strictly incoherent (classical) operations (SIOs)}
If we focus on models generated solely through classical memory states and channels, the topological complexity increases. Specifically, the presence of coherences (relative to the the measurement basis) can offer a quadratic advantage~in the Hilbert space dimension. Classically generated processes can also be expressed as QHMMs where the input and output states of the model are incoherent (diagonal in the measurement basis) and the maps $\mathcal{E}^x$ in Eq.~\eqref{eq:ProbAsSequenceOfInstruments} are strictly incoherent operations (SIOs) with respect to the same basis~\cite{Plenio_coherence_2014, winter_operational_coherence_2016,Vedral_Coherence_2016}.  The operations reduce to the dynamics of an edge-emitting HMM, yielding a special case of $c_Q(\overrightarrow{X})$, the classical topological complexity, $c_C(\overrightarrow{X})$, which aligns with Eq.~\eqref{eq:TCgQDef} under the constraint of classical memory.
\begin{theorem}\label{theorem:classical tc bound}
        The classical topological complexity $c_C({\overrightarrow{X}})$ is bounded from below by the ceiling of the number of distinct, non-zero eigenvalues in the spectrum of any valid transfer operator of the process raised to $\frac{1}{2}$, i.e., 
        \begin{align}\label{eq:cCInequality}
            c_C(\overrightarrow{X})&\geq \log\ceil{|\Lambda_{\overrightarrow{X}}|^\frac{1}{2}}, \text{and}\\
            c_C(\overrightarrow{X})&\geq c_Q(\overrightarrow{X}).
        \end{align}
\end{theorem}
\begin{proof} 
Consider the QHMM $R_C$ which minimizes Eq.~\eqref{eq:cQInequality} subject to the constraint that the model is composed of maps $\mathcal{E}^{(x)}_{R_C}$ with $\dim(\mathcal{H}^{(in)}_{R_C}) = m$. SIOs yield vectorized operators that are $m^4\times m^4$ matrices, but with at most $m^2\times m^2$ non-zero entries and $\rank(\mathbf{E}_{R_C,R_C})\leq m^2$. Importantly, SIO operations preserve diagonal states in the incoherent basis, which imposes structural constraints on the Kraus operators and effectively limits the support of their action. Thus the dimension of the transfer operator for any minimal classical QHMM satisfies $|\Lambda_{\overrightarrow{X}}| \leq m^2$. Moreover, any minimal QHMM (without the SIO constraint) cannot exceed the classical HMM's memory since this would contradict Thm.~\eqref{theorem:quantum tc bound}. See Secs.~\ref{app: classical HMMs and SIOs} and \ref{app:proof of classical tc bound} in the Appendix for details.
\end{proof}

Theorem~\eqref{theorem:classical tc bound} shows that a process $\overrightarrow{X}$ with classical topological complexity $c_C(\overrightarrow{X})$, can have a lower quantum topological complexity $c_Q(\overrightarrow{X})<c_C(\overrightarrow{X})$ indicating a quantum memory advantage. This occurs if and only if the maps $\{\mathcal{E}^{(x)}_{R_Q}\}$ of the minimal QHMM for $\overrightarrow{X}$ create or utilize coherence. The next section provides a constructive example of this advantage. {While examples of quantum generative models outperforming classical ones in memory size have appeared in previous work~\cite{fanizza2024quantum}, they typically rely on specially constructed process families. In contrast, the following example presented in Observation.~\ref{observation: proof of existence}, demonstrates a strict quantum advantage by explicitly violating the classical spectral bound established in Thm.~\ref{theorem:classical tc bound}. This provides a verifiable and concrete instance of quantum advantage in memory dimension relative to classical models (i.e., those constrained to incoherent operations).}
\begin{figure}
    \includegraphics[scale=0.7]{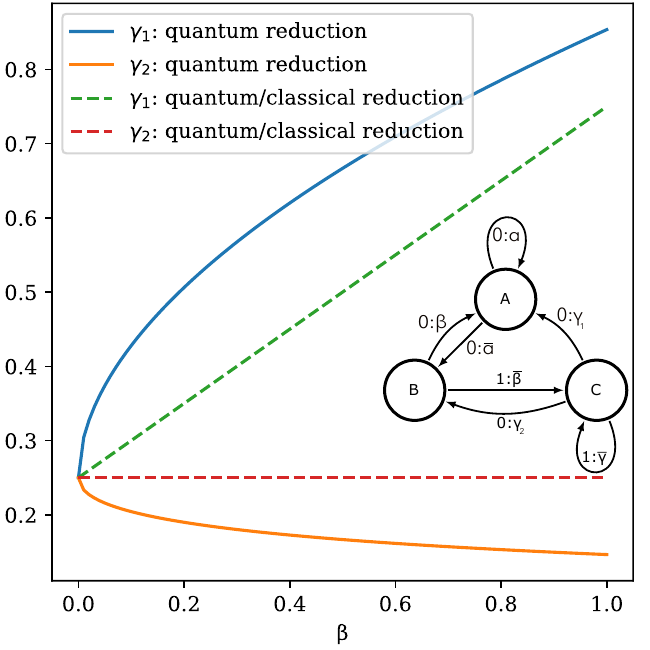}
    \caption{The inset shows a classically irreducible three-state, two parameter ($\alpha,\beta$) generative model that emits from an alphabet $\{0,1\}$. The model demonstrates a quantum advantage in the number of generative states for certain transition values. We define {$\bar{\alpha}:=1-\alpha$, $\bar{\beta}:=1-\beta$, $\bar{\gamma}:=1-\gamma_1-\gamma_2$} and set $\gamma_1=\left(\frac{\sqrt{\alpha}+\sqrt{\beta}}{\nu} \right)^2$, ${\gamma_2} =\frac{1}{\nu^2} {\bar{\alpha}}$, and $\bar{\gamma} = \frac{1}{\nu^2} {\bar{\beta}}$ where the normalisation constant $\nu$ is determined implicitly by ${\gamma_1}+{\gamma_2}+\bar{\gamma}=1$. We set $\alpha=0.5$ and show the curves of $\{\gamma_1,\gamma_2,\bar{\gamma}\}$ that lead to a reduction in the the quantum regime. We also obtain a classical reduction for the parameter values shown in the dashed lines.}
    \label{fig:exampleModel}
\end{figure}
\\\\
{\textbf{Observation~1.}}~\label{observation: proof of existence} \textit{There exist processes $\overrightarrow{X}$ for which $c_Q(\overrightarrow{X})< c_C(\overrightarrow{X})$.}
\begin{proof}
\textit{Quantum generative advantage}
When the maps $\{\mathcal{E}^{(x)}\}$ in Eq.~\eqref{eq:ProbAsSequenceOfInstruments} are SIOs and the initial state $\rho$ is diagonal in the same basis, the QHMM reduces to a classical HMM with the mappings: $\{\mathcal{E}^{(x)}\}\mapsto \{T^{(x)}\}$ and $\rho\mapsto \pi_\rho$, where $\{T^{(x)}\}$ are symbol-labeled transition matrices and $\pi_\rho$ is a vector describing the distribution over internal states given by the diagonal of $\rho$. The matrices $\{T^{(x)}\}$ update classical memory states from $r$ to $r'$ with probability $T^{(x)}_{r \rightarrow r'}$ upon emitting $x$ (see Sec.~\ref{app: calculating probs from HMMs} of the Appendix for details). A more general QHMM generating $\overrightarrow{X}$ can also be described (non-uniquely) using classical HMMs by re-defining quantum memory states that, unlike their classical counterparts, are not necessarily orthogonal. This quantum coherence enables memory advantages, extensively studied in unifilar HMMs \cite{mahoney_occams_2016,binder_practical_2018,elliott_Extreme_2020}.  Sec.~\ref{app: Quantum Circuit to Implement any Classical HMM} of the Appendix provides a description for translating classical HMMs into QHMMs, achieving memory reduction when memory states exhibit linear dependence. \\\\
Now let us consider the 3-state model in Fig.~\ref{fig:exampleModel} with two emission symbols such that the set of transition matrices $\{T^{(0)},T^{(1)}\}$ is 
\begin{align}\label{eq:transitionMatrices}
T^{(0)} = \begin{bmatrix}
        \alpha & \beta & \gamma_1\\
        \bar{\alpha} & 0 & \gamma_2\\
        0 & 0 & 0
    \end{bmatrix}~;~~
T^{(1)} = \begin{bmatrix}
        0 & 0 & 0 \\
        0 & 0 & 0 \\
        0 & \bar{\beta} & \bar{\gamma} \\
    \end{bmatrix} 
\end{align}
with {$\bar{\alpha}:=1-\alpha$, $\bar{\beta}:=1-\beta$ and, $\bar{\gamma}:=1-\gamma_1-\gamma_2$}. 
This HMM defines a stochastic process where word probabilities are calculated by applying the transition matrices sequentially to an initial state distribution (see Sec.~\ref{app: calculating probs from HMMs} of the Appendix for details). {Alternatively, we may consider a set of quantum states $\{\ket{\sigma_A},\ket{\sigma_B},\ket{\sigma_C}\}$ defined as
\begin{align}\label{eq:quantum gen states}
    \ket{\sigma_j}:=\sum_{k}e^{i\lambda(j,k,x)}\sqrt{T_{j\to k}^{(x)}}\ket{k}\!\!\ket{x}
\end{align}
where $\ket{k}\in\{\ket{A},\ket{B},\ket{C}\}$ and $\ket{x}\in \{\ket{0},\ket{1}\}$ are orthonormal bases corresponding to the internal states and observables of the HMM respectively, $e^{i\lambda(j,k,x)}$ is a phase term, and the elements $\sqrt{T^{(x)}_{j\to k}}$ are drawn from the transition matrices in Eq.~\eqref{eq:transitionMatrices}. Note that $\{\ket{\sigma_k}\}$ are not necessarily orthogonal. These lead to the Kraus operators of each $\mathcal{E}^{(x)}$ 
\begin{align}\label{eq: KrausOPS for example}
    K_{k}^{(x)} := \ket{\sigma_k}\!\!\bra{k}\!\!\bra{x}.
\end{align} 
Together with analogous initial states~\footnote{{The analogous quantum state for initial classical distribution $\pi=(\Pr(A),\Pr(B)),\Pr(C)$ is given by $\rho=\sum_k\Pr(k)\dyad{\sigma_k}$.}}, the QHMM defined by these Kraus operators, gives rise to the same process as the corresponding HMM. The memory size of this QHMM is given by the log of the dimension of the input/output space of the maps $\{\mathcal{E}^{(x)}\}$, i.e., the span of the states $\{\ket{\sigma_k}\}$. With this in mind, we consider the case where} $\gamma_1=\left(\frac{\sqrt{\alpha}+\sqrt{\beta}}{\nu} \right)^2$, ${\gamma_2} =\frac{1}{\nu^2} {\bar{\alpha}}$, and $\bar{\gamma} = \frac{1}{\nu^2} {\bar{\beta}}$. In this case we clearly have $\ket{\sigma_C} = \frac{1}{\nu}(\ket{\sigma_A}+\ket{\sigma_B})$. The normalisation constant $\nu$ satisfies ${\gamma_1}+{\gamma_2}+\bar{\gamma}=1$. Since $\ket{\sigma_C}$ is a linear combination of $\ket{\sigma_A}$ and $\ket{\sigma_B}$ {we have $\text{span}(\{\ket{\sigma_k}\})=2$}, and thus the process requires only a two-dimensional quantum memory, i.e., $c_Q(\overrightarrow{X})=\log(2)$. {See Appendix.~\ref{app: Quantum Circuit to Implement any Classical HMM} for details.} However using Thm.~\eqref{theorem:classical tc bound} we can verify $c_C(\overrightarrow{X})=\log(3)$. {That is, the presentation of the HMM given by Eq.~\eqref{eq:transitionMatrices} is already minimal in the number of states,} except for trivial cases where $\alpha=0$ and/or $\beta=0$. In those limits, the effective number of classical states drops below three, eliminating the quantum advantage. However, in general, this construction shows that quantum models can exploit coherence and linear dependence to reduce memory requirements. In particular, memory reduction can occur in two ways: reductions present in both classical and quantum cases, or purely quantum reductions due to coherence or phase enhancement. In the above model, the classical reduction line corresponds to $\gamma_1 = \frac{\alpha+\beta}{\nu}$, $\gamma_2 = \frac{1-\alpha}{\nu}$  and $\bar{\gamma} = 1- \gamma_1-\gamma_2$ with normalization factor $\nu = \alpha+\beta+\bar{\beta}-\bar{\alpha} = 2$. {These cases are highlighted in Fig.~\ref{fig:exampleModel}}.
\end{proof}
This example clearly demonstrates a quantum violation of the classical bound and a quantum advantage in the number of generating memory states. {This construction is not unique, but we highlight that it readily generalizes to higher-dimensional classical models, by encoding larger sets of classically distinct memory states as non-orthogonal quantum states. This explicitly applies to operations that go beyond SIOs which, as highlighted by Theorems \eqref{theorem:quantum tc bound} and \eqref{theorem:classical tc bound}, enable representation within a lower-dimensional Hilbert space that can lead to reductions in generative memory, even when the classical memory requirements remain fixed.}

\section{Conclusions} In this paper, we have derived new bounds on the memory required to sequentially generate a stochastic process, i.e., its quantum generative topological complexity. By analysing the spectral equivalence of transfer operators, we set firm lower bounds on the memory necessary to generate a process via any QHMM. Subsequently, we demonstrate that when restricting operations to be strictly incoherent, we recover classical hidden Markov models (HMMs). We find that the bound on the minimum Hilbert space dimension of the memory is quadratically larger than that of their quantum counterpart. This shows that any violation of the classical topological complexity bound requires quantum coherence. We illustrate this with an example, showing that memory advantages are possible with QHMMs, where operations are not strictly incoherent. Such processes are highly non-unique, allowing the construction of arbitrarily many models by defining quantum memory states as linear combinations of each other. These models can be analyzed via Thms.~\eqref{theorem:quantum tc bound} and~\eqref{theorem:classical tc bound}. Our results demonstrate the spectral bound's utility in identifying generative advantages and highlight non-SIO operations as a resource for stochastic process generation. This parallels phase-based resource advantage in similar contexts~\cite{Liu_optimal_2019,Glasser_Expressive_2019}). 
{This work opens several directions for future investigation. One avenue is the exploration of broader classes of processes that exhibit quantum memory advantages, as well as understanding whether such advantages persist in the setting of approximate simulation.  We also note that the phase degrees of freedom in Eq.~\ref{eq:quantum gen states} may be exploited to engineer linear dependencies between the generating memory states of a QHMM. How to do so optimally remains an open question, paralleling similar challenges in the predictive setting~\cite{Liu_optimal_2019}. More generally, connections to sequential quantum measurements and resource theories of coherence also offer promising ground for further study, especially given the fact that monitored open quantum systems can be understood in terms of sequentially applied quantum instruments or indeed QHMMs~\cite{Wiseman2010}.}

\section{Acknowledgments} The authors would like to thank Thomas Elliott, Paul Riechers for insightful comments. The research conducted in this publication was funded by the Irish Research Council under grant number IRCLA/2022/3922.

\bibliographystyle{apsrev4-2}
\bibliography{mybib}

\clearpage
\appendix
\begin{onecolumngrid}

\section{Proof of Theorem 1}\label{sec:ProofofLemma1}
If two models generate the same process it must hold that
\begin{align}\label{eq:rr=qqVecotrisedApp}
    \langle\langle \mathds{1} | \mathbf{E}_{RR}^L|\rho_{RR}\rangle\rangle= \langle\langle \mathds{1} | \mathbf{E}_{QQ}^L|\rho_{QQ}\rangle\rangle \forall L.
\end{align}
Assuming the matrix $\mathbf{E}_{RR}$ to be diagonalizable, it can be written in its eigendecomposition
\begin{align}
    \mathbf{E}_{RR}=\sum_{\lambda_R}\lambda_R\Pi_{\lambda_R}, 
\end{align}
where $\lambda_R$ are the eigenvalues and $\Pi_{\lambda_R}$ are the corresponding eigenprojectors. Then,
\begin{align}
    \langle\langle \mathds{1} | \mathbf{E}_{RR}^L|\rho_{RR}\rangle\rangle&=  \langle\langle \mathds{1} |\sum_{\lambda_R}\lambda_R^L\Pi_{\lambda_R}|\rho_{RR}\rangle\rangle\\
&=\sum_{\lambda_R}\lambda_R^L\alpha_{\lambda_R}
\end{align}
where $\alpha_{\lambda_R}=\langle\langle\mathds{1} |\Pi_{\lambda_R}|\rho_{RR}\rangle\rangle$ and $\Pi_{\lambda_R}^L = \Pi_{\lambda_R}~\forall~L$. Then, the equality in Eq.~\eqref{eq:rr=qqVecotrisedApp} can equivalently be expressed as
\begin{align}\label{eqn:limProof1}
\sum_{\lambda_R}\lambda^L_R\alpha_{\lambda_R}&=\sum_{\lambda_{Q}}\lambda^L_{Q}\alpha_{\lambda_{Q}},
\end{align}
where $\lambda_R$ and $\lambda_{Q}$ are the eigenvalues of $\mathbf{E}_{RR}$ and $\mathbf{E}_{QQ}$ respectively. To proceed we introduce
\begin{align}
    {\Lambda_R=\{ \lambda_{R}~|~\exists~\Pi_{\lambda_R} 
    \neq 0 : \mathbf{E}_{RR}\Pi_{\lambda_R} = \lambda_R \Pi_{\lambda_R} \}}
\end{align}
and 
\begin{align}
    {\Lambda_Q=\{ \lambda_{Q}~|~\exists~\Pi_{\lambda_Q} 
    \neq 0 : \mathbf{E}_{QQ}\Pi_{\lambda_Q} = \lambda_Q \Pi_{\lambda_Q} \}},
\end{align}
the sets of distinct, non-zero spectrum of the transfer operators $\mathbf{E}_{RR}$ and $\mathbf{E}_{QQ}$ respectively. By construction, all eigenvalues in $\Lambda_Q$ and $\Lambda_R$ only appear once in the set. Next, subtracting the right hand side of Eq.~\eqref{eqn:limProof1} from the left hand side we have
\begin{align}
    \sum_{\lambda_R\in\Lambda_R}\lambda_R^L\alpha_{\lambda_R}- \sum_{\lambda_Q\in\Lambda_Q}\lambda_Q^L\alpha_{\lambda_Q} = 0.
\end{align}
By introducing $\Lambda_{RQ}:=\Lambda_R\cup\Lambda_Q$ this becomes
\begin{align}\label{eq:eigenvaluesTo0}
\sum_{\lambda_{RQ}\in\Lambda_{RQ}}\lambda_{RQ}^L\alpha_{\lambda_{RQ}}  = 0,
\end{align}
where $\alpha_{\lambda_{RQ}} := \alpha_{\lambda_R=\lambda_{RQ}} - \alpha_{\lambda_Q=\lambda_{RQ}}$ are the differences between the projectors corresponding to eigenvalues $\lambda_{RQ}$. \\

\noindent 
We can re-write the system of equations in Eq.~\eqref{eq:eigenvaluesTo0} from $L=1$ to $L=d$ (where $d$ is the size of $\Lambda_{RQ}$) as, 
\begin{align}
V(\Lambda_{RQ})^TD(\Lambda_{RQ})\overrightarrow{\alpha}_{\lambda_{RQ}} = 0
\end{align}
where $D(\Lambda_{RQ})$ is the diagonal matrix
\begin{align}
    D(\Lambda_{RQ}) \equiv 
\begin{bmatrix}
& \lambda_{RQ}^{(1)}& 0 & 0 & \cdots & 0
\\ 
& 0 & \lambda_{RQ}^{(2)} & 0& \cdots & 0
\\
& 0& 0 & \lambda_{RQ}^{(3)} & \cdots & 0
\\
& \vdots & \vdots &\vdots & \cdots& \vdots
\\ 
& 0& 0 & 0 & \cdots & \lambda_{RQ}^{(d)}
\end{bmatrix},
\end{align}
 $V(\Lambda_{RQ})$ is the Vandermonde matrix~\cite{kailath_displacement_1995},
\begin{align}
    V(\Lambda_{RQ}) \equiv 
\begin{bmatrix}
& 1& \lambda_{RQ}^{(1)} & (\lambda_{RQ}^{(1)})^2 & \cdots & (\lambda_{RQ}^{(1)})^{d-1}
\\ 
& 1 & \lambda_{RQ}^{(2)} & (\lambda_{RQ}^{(2)})^2 & \cdots & (\lambda_{RQ}^{(2)})^{d-1}
\\
& 1& \lambda_{RQ}^{(3)} & (\lambda_{RQ}^{(3)})^2 & \cdots & (\lambda_{RQ}^{(3)})^{d-1}
\\
& \vdots & \vdots &\vdots & \cdots& \vdots
\\ 
& 1& \lambda_{RQ}^{(d)} & (\lambda_{RQ}^{(d)})^2 & \cdots & (\lambda_{RQ}^{(d)})^{d-1}
\end{bmatrix},
\end{align}
and the vector $\overrightarrow{\alpha}_{\lambda_{RQ}}$ is
\begin{align}
    \overrightarrow{\alpha}_{\lambda_{RQ}} = \begin{bmatrix}
    \alpha_{\lambda_{RQ}^{(1)}}\\
    \alpha_{\lambda_{RQ}^{(2)}}\\
    \alpha_{\lambda_{RQ}^{(3)}}\\
    \vdots\\
    \alpha_{\lambda_{RQ}^{(d)}}
    \end{bmatrix}.
\end{align}
Then the determinant of $V(\Lambda_{RQ})$ is given by the general form of the Vandermonde determinant
\begin{align}
\text{det}(V(\Lambda_{RQ}))= \prod_{1 \leq i < j \leq (d-1)} (\lambda^{(j)}_{RQ}-\lambda^{(i)}_{RQ}).
\end{align}
Since all elements of $\Lambda_{RQ}$ are distinct and nonzero, this product is also nonzero and thus Eq.~\eqref{eq:eigenvaluesTo0} implies $\alpha_{\lambda_{RQ}}=0$ $\forall~\lambda_{RQ}$. \\

\noindent So, we have three cases:
\begin{enumerate}
    \item The eigenvalue $\lambda_{RQ}$ is in both $\Lambda_R$ and $\Lambda_Q$, i.e, $\lambda_{RQ}\in\Lambda_R\cap\Lambda_Q$. Then,
    \begin{align}
    \alpha_{\lambda_{RQ}} = 0 \implies \alpha_{\lambda_R=\lambda_{RQ}} = \alpha_{\lambda_Q=\lambda_{RQ}}.    
    \end{align}
    In other words, the coefficients $\alpha_{\lambda}$ corresponding to the eigenvalue $\lambda_{RQ}$ are the same for both models $R$ and~$Q$.
    \item  The eigenvalue $\lambda_{RQ}$ is in $\Lambda_R$ but not $\Lambda_Q$, i.e., $\lambda_{RQ}\in\Lambda_R~~\text{but}~~\notin\Lambda_Q$. Then, 
    \begin{align}
    \alpha_{\lambda_{RQ}} &= 0 \\\implies \alpha_{\lambda_{RQ}} &= \alpha_{\lambda_R=\lambda_{RQ}}=0.
    \end{align}
    \item The eigenvalue $\lambda_{RQ}$ is in $\Lambda_Q$ but not $\Lambda_R$, i.e., $\lambda_{RQ}\in\Lambda_Q~~\text{but}~~\notin\Lambda_R$. Then, 
    \begin{align}
    \alpha_{\lambda_{RQ}} &= 0 \\\implies \alpha_{\lambda_{RQ}} &= \alpha_{\lambda_Q=\lambda_{RQ}}=0.    
    \end{align}
\end{enumerate}
Therefore, whenever a non-zero eigenvalue appears in only one set (either $\Lambda_{R}$ or $\Lambda_{Q}$), the corresponding coefficients must vanish. {Since $\Lambda_R$ and $\Lambda_Q$ are defined such that they exclude eigenvalues for which $\Pi_\lambda=0$ they must also exclude terms for $\alpha_\lambda=0$. Then by (2) and (3) above, eigenvalues $\lambda_{RQ}$ which are not in $\Lambda_{R}$ are necessarily also not in $\Lambda_{Q}$ (and vice versa). We thus conclude $\Lambda_{RQ}=\Lambda_{R}=\Lambda_Q$.}\\\\
{
We note that in the non-diagonalisable case, we instead revert to $\mathbf{E}_{RR}$'s Jordan decomposition,
\begin{align}
    \mathbf{E}_{RR} = \sum_{\lambda_R} ({\lambda_R} \mathbb{I} + \mathcal{N}_{\lambda_R})\Pi_{\lambda_R},
\end{align}
where $\Pi_{\lambda_R}$ projects onto the generalized eigenspace of ${\lambda_R}$, and $\mathcal{N}_{\lambda_R}$ is a nilpotent operator with $\mathcal{N}_{\lambda_R}^{k_{\lambda_R}} = 0$. We note that in this case that for large $L$, $L\geq k_{\lambda_R}$, we have,
\begin{align}
    \mathbf{E}_{RR}^L = \sum_{\lambda_R}\lambda_R^L\Pi_{\lambda_R}
\end{align}
since the nilpotent terms satisfy $\mathcal{N}_{\lambda_R}^{L\geq k_{\lambda_R}}=0$. The rest of the proof remains unchanged and we conclude that we must have equivalent spectra of transfer operators also in the non-diagonalizable case.}

\section{Calculating sequence probabilities from classical HMMs}\label{app: calculating probs from HMMs}
\noindent A classical HMM $M$ is defined by a set of sub-stochastic transition matrices $\{T^{(x)}\}$ and an initial distribution over internal states $\{r_i\}$, $\pi=(\Pr(r_1),P(r_2),\dots)^T$:
\begin{align}\label{eq: classical HMM}
M=(\{T^{(x)}\},\pi),   
\end{align}
where the elements of the transition matrices are
\begin{align}
    T_{i\to j}^{(x)} = \text{Pr}(x,i|j),
\end{align}
i.e., the conditional probability of the model emitting the symbol $x$ and transitioning from memory states $j$, given it was in state $i$. The probability of any sequence of outputs is then given by
\begin{align}
    \Pr(x_{0:L})_{\pi} =\overrightarrow{1} \prod_{i=0}^L T^{(x_i)}\pi,
\end{align}
where $\overrightarrow{1}$ is the row vector of all $1$s. Each transition matrix has elements
\begin{align}
(T^{(x)})_{ij} = \Pr(x,j|i)
\end{align}
describing the probability of transitioning from internal state $r_i$ to $r_j$ when the process output is $x$. The transition matrix for a word ${x_{0:L}}$ can also be expressed in terms of these matrices as
\begin{align}
    T^{(x_{0:L})} = \prod_{i=0}^LT^{(x_i)},
\end{align}
which acts on a distribution $\pi$ over orthogonal internal (memory) states $\{r_i\}$ to update to the conditional memory state $\pi_L^{(x_{0:L})} = T^{x_{0:L}}\pi$. Then, the probability of the word $(x_{0:L})$ given the initial state $\pi$ is given by, 
\begin{align}\label{eq:p(x0:l)ClassicalHMM}
    \Pr(x_{0:L})_{\pi} = \overrightarrow{1}T^{(x_{0:L})}\pi .
\end{align}

\section{Classical HMMs and strictly incoherent operations}\label{app: classical HMMs and SIOs}

We can identify classical HMMs within the larger class of QHMMs by selecting any computational memory basis $\mathcal{M}=\{|m \rangle \}$ and disallowing quantum coherence in that basis. In this case, the quantum instrument $\mathcal{E}$ can be viewed as a measurement in that basis, followed by a stochastic transformation on that memory state $|i\rangle$ that produces the output $x$ and maps the memory to state $|j \rangle $ with some probability $\Pr(x,j|i)$.  Thus, the action of elements of any such instrument can be expressed as
\begin{align}\label{eq:Kraus operators for SIO}
    \mathcal{E}^{x}(\rho) = \sum_{i,j}\Pr(x,j|i)|j \rangle\langle i |\rho  |i \rangle \langle j|.
\end{align}
Note from the previous appendix that this transformation explicitly encodes the dynamics of an edge-emitting HMM $(T^{(x)}_M)_{ij}= \Pr(x,j|i)$.  Expressing such an instrument in terms of Kraus operators requires their magnitude to be proportional to the square-root of the transition weights and allows for the choice of a complex phase as (irrelevant) degree of freedom:
\begin{align}
    K^{(x)}_{i \rightarrow j} = \sqrt{\Pr(x,j|i)}e^{\iota \phi (x,i,j)}|j \rangle \langle i|,
\end{align}
where $\iota$ is the imaginary unit. Quantum channels of this form are known as strictly incoherent operations (SIOs), which neither create nor use coherence w.r.t. $\mathcal M$~\cite{Plenio_coherence_2014, winter_operational_coherence_2016,Vedral_Coherence_2016}.

The transfer operator of an SIO in A-form is:
\begin{align}
    \mathbf{E}_{R_C}= \sum_{i,i',j,j',k,k',l,l'}|ii'jj' \rangle \langle kk'll'| \sum_{x}\langle k|\mathcal{E}^x(|i \rangle \langle i'|)\\
    |k'\rangle \langle l|\mathcal{E}^x(|j \rangle \langle j'|)|l'\rangle
\end{align}
Due to the structure of the SIO, the matrix representation of the transfer operator is highly sparse:
\begin{align}
(\mathbf{E}_{R_C})_{ii'jj'}^{kk'll'} & \equiv \sum_{x}\langle k|\mathcal{E}^x(|i \rangle \langle i'|)|k'\rangle \langle l|\mathcal{E}^x(|j \rangle \langle j'|)|l'\rangle
\\ & =\sum_{x}\Pr(x,k|i)\Pr(x,l|j) \delta_{i,i'}\delta_{j,j'}\delta_{k,k'}\delta_{l,l'}.
\end{align}
Here, $(E_{R_C})_{ijkl}=\sum_{x}\Pr(x,k|i)\Pr(x,l|j)$ are precisely the elements of the matrix representation of the classical transfer operator $E_{R_C}$ \cite{yang_Measures_2020}. Therefore, we can express the quantum transfer operator as a classical operator embedded in a larger space
\begin{align}
\mathbf{E}_{R_C}= \sum_{ijkl}(E_{R_C})_{ij}^{kl}|iijj\rangle \langle kkll|.
\end{align}
In this way, we directly connect the quantum transfer operator to the complexity of classical models.

\begin{figure}
    \centering
    \includegraphics[width=0.75\linewidth]{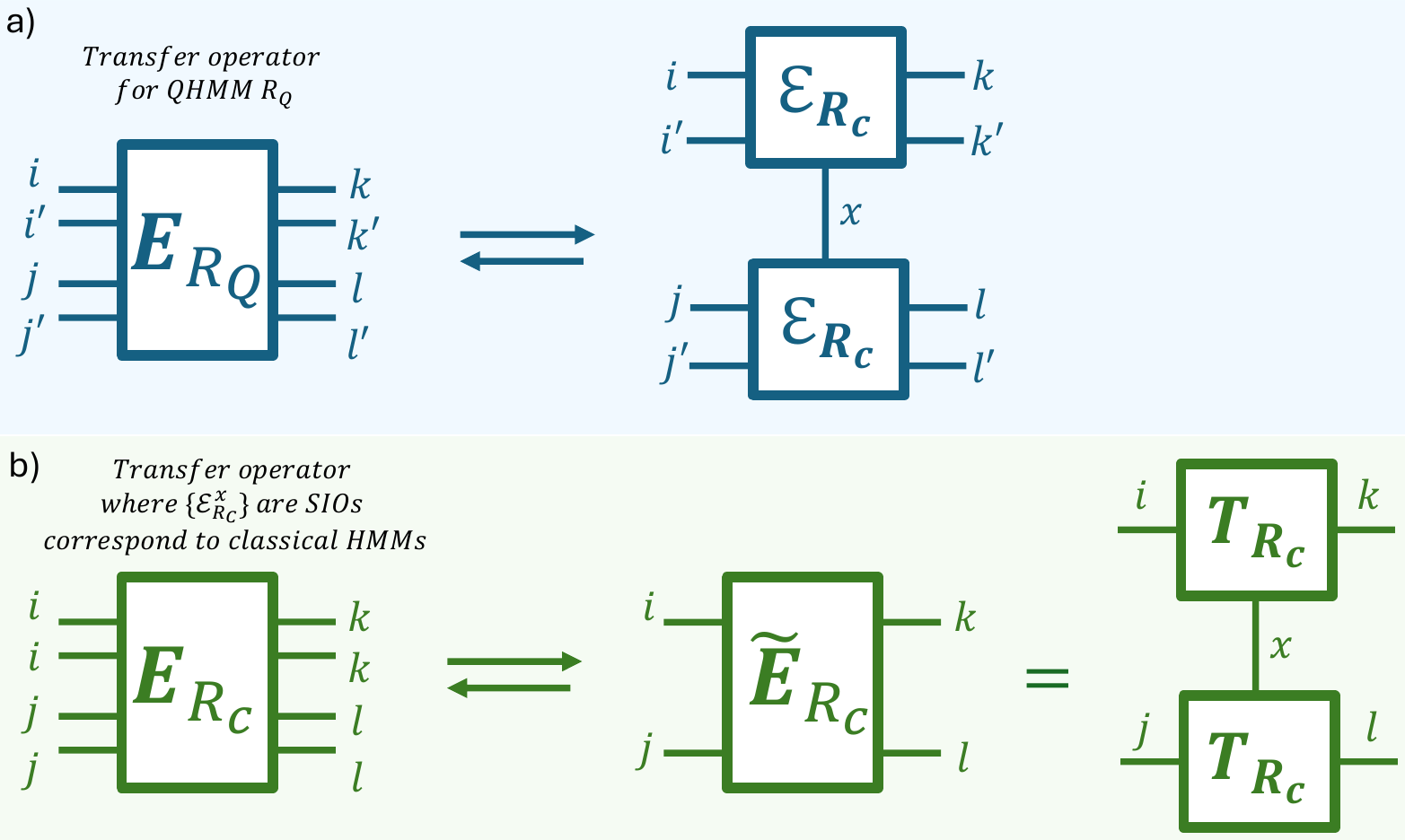}
    \caption{a) The transfer operator of a QHMM can be obtained as $\mathbb{E}_{R_Q} = \sum_{x}\mathcal{E}_{R_C}\otimes \mathcal{E}_{R_C}$ and is an eight index object, with four input ($i,i'j,j'$) and four output indices ($k,k'l,l'$). On the other hand we can also define a transfer operator for QHMMs with only strictly incoherent operations. Each set of doubles of input indices $\{i,i'\}$ and $\{j,j'\}$ correspond to the input indices of density operators and likewise for the output spaces $\{j,j'\}$ and $\{l,l'\}$ of the maps $\mathcal{E}_{R_Q}$. b) If instead, operations are restricted to be strictly incoherent, then the resulting transfer operator is constructed out of maps from single index objects to single index objects, i.e., from diagonal states to diagonal states. This can equivalently be expressed as $\mathbb{E}_{R_C} = \sum_{x}T^{(x)}\otimes T^{(x)}$, where $\{T^{(x)}\}$ are transition matrices of classical HMMs.}
    \label{fig:enter-label}
\end{figure}

\section{Proof of Theorems 2 and 3}\label{app:proof of classical tc bound}

Consider that minimal QHMMs $R_{min}$ of a process $\overrightarrow{X}$ act on an $n$-dimensional memory, i.e., $\dim(\mathcal{H}^{in}_{R_{min}})=n$ where the maps corresponding to a given model are $\{\mathcal{E}_{R_{min}}^{(x)}:\mathcal{L}(\mathcal{H}^{in}_{R_{min}})\mapsto \mathcal{L}(\mathcal{H}^{in}_{R_{min}}))\}$, where $\mathcal{L}(\mathcal{H})$ denotes linear operators on the Hilbert space $\mathcal{H}$. The minimality of $R_{min}$ implies that there is no alternative model $R$ of $\overrightarrow{X}$ with $\dim(\mathcal{H}^{in}_{R})<n$. The vectorised maps $\{\tilde{\mathcal{E}}_{R_{min}}^{(x)}\}$ are $n^2\times n^2$ matrices and the corresponding (vectorised) transfer operator $\mathbf{E}_{R_{min}}$ is an $n^4\times n^4$ matrix with distinct nonzero spectrum $\Lambda_{\overrightarrow{X}}$. In general $\mathbf{E}_{R_{min}}$ is the matrix
\begin{align}
    \mathbf{E}_{R_{min}}&= \sum_{i,i',j,j',k,k',l,l'}(\mathbf{E}_{R_{min}})_{ii'jj'}^{kk'll'}\dyad{ii'jj'}{kk'll'}
\end{align}
where $\{\ket{ii'jj'}\}$ and $\{\ket{kk'll'}\}$ are bases of the input and output spaces of $\mathbf{E}_{R_{min}}$ respectively and the size of the alphabet for each of the input space is precisely $\dim(\mathcal{H}^{(in)})$. The number of eigenvalues $\{\lambda_{R_{min}}\}$ of $\mathbf{E}_{R_{min}}$ will be equal to its dimensions and lower bounded by $|\Lambda_{\overrightarrow{X}}|$ (which only includes unique eigenvalues). In other words, since $c_Q(\overrightarrow{X})=\log n$ we have
\begin{align}
 &c_Q(\overrightarrow{X})^4=\log|\{\lambda_{R_{min}}\}|\geq \log|\Lambda_{\overrightarrow{X}}|   \\
 \implies&c_Q(\overrightarrow{X})=\log|\{\lambda_{R_{min}}\}|^{1/4} \geq \log\ceil{|\Lambda_{\overrightarrow{X}}|^{1/4}}
\end{align}
with equality holding if and only if $\{\lambda_{R_{min}}\}\equiv \Lambda_{\overrightarrow{X}}$, i.e., the multiplicity of all $\lambda_{R_{min}}$ is 1 and $\lambda_{R_{min}}\neq 0~\forall~\lambda_{R_{min}}$.\\\\
On the other hand, we can find an alternative minimal classical model $R_{Cmin}$ of the same process $\overrightarrow{X}$ which acts on an $m$-dimension classical memory, i.e., $\dim(\mathcal{H}^{in}_{R_{Cmin}})=m$ where the maps corresponding to the model are $\{\mathcal{E}_{R_{Cmin}}^{(x)}:\mathcal{L}(\mathcal{H}^{in}_{R_{Cmin}})\mapsto \mathcal{L}(\mathcal{H}^{in}_{R_{Cmin}})\}$. In order for the transformation of states to be fully classical we require that the input and output states are diagonal in the measurement basis and that transformations $\{\mathcal{E}_{R_{min}}^{(x)}\}$ are SIOs with respect to some computational basis. Under this construction, coherent terms on both input and output spaces of the transfer operator must be zero such that the vectorised transfer operator $\mathbf{E}_{R_{Cmin}}$ is
\begin{align}\label{eq:ERcmin}
        \mathbf{E}_{R_{Cmin}} &= \sum_{i,i',j,j',k,k',l,l'}(\mathbf{E}_{R_{Cmin}})_{ii'jj'}^{kk'll'}\dyad{ii'jj'}{kk'll'}\delta_{i,i'}\delta_{j,j'}\delta_{k,k'}\delta_{l,l'}
    \\&=\sum_{i,j,k,l}(E_{R_{Cmin}})_{ij}^{kl}\dyad{iijj}{kkll},
\end{align}
where $E_{R_{Cmin}}$ is a classical transfer operator. This means that if $\dim(\mathcal{H}^{(in)\otimes 2}_{R_{C{min}}})=m$ where only strictly incoherent operations act on the memory states, $\mathbf{E}_{R_{Cmin}}$ is an $m^4\times m^4$ matrix, but only has $m^2\times m^2$ non-zero terms. Importantly, the structure of the operator guarantees that the nullity (dimension of the kernel) of the operator is
\begin{align}
 \text{nullity}(\mathbf{E}_{R_{Cmin}})\geq m^4-m^2   
\end{align}
and conversely that the rank is 
\begin{align}
    \rank(\mathbf{E}_{R_{Cmin}})\leq m^2.    
\end{align}
Then the set $\Lambda_{\overrightarrow{X}}$, which by Thm.~1 of the main text is the same for any valid model of the system, is related to the size of the minimal generating memory $\log m=c_C(\overrightarrow{X})$ as
\begin{align}
    &c_C(\overrightarrow{X})^2\geq \log[\rank({\mathbf{E}_{R_{Cmin}}})] \geq \log|\Lambda_{\overrightarrow{X}}|\\
    \implies&c_C(\overrightarrow{X})\geq \log\ceil{|\Lambda_{\overrightarrow{X}}|^{1/2}},
\end{align}
where equality holds if and only if the eigenvalues of the remaining spaces spanned by $\{\dyad{iijj}{kkll}\}$ are non-zero and distinct. Moreover, we note that the minimal memory for QHMMs must be lower-bounded by the minimal memory for classical HMMs. 

\section{QHMM corresponding to a classical HMM}
\label{app: Quantum Circuit to Implement any Classical HMM}
\noindent Given a classical HMM $(\{T^{(x)}\},\pi)$ as per Eq.~\eqref{eq: classical HMM}, which generates $\overrightarrow{X}$, we seek a QHMM $(\{\mathcal{E}^{x}\},\rho)$ that utilises coherence between memory states to achieve a memory advantage over its classical counterpart. In particular, consider the classical model which generates word sequences $x_{0:L}$ with probabilities given by Eq.~\eqref{eq:p(x0:l)ClassicalHMM}, 
\begin{align}
    \text{Pr}(x_{0:L})_{\pi} =\overrightarrow{1}^TT^{(x_{L})}T^{(x_{L-1})}\dots T^{(x_{0})}\pi.  
\end{align}
where the transition matrices have elements $T_{i\to j}^{(x)}=\text{Pr}(x,j|i)$, describing the probability of transitions from internal state $i$ to $j$, when the output $x$ is observed. Then, in order to construct the QHMM which executes the model, we require CP maps $\{\mathcal{E}^{(x)}\}$ satisfying 
\begin{align}\label{eq:prbabilties in terms of instruments}
    \text{Pr}(x_{0:L})_{\rho_m} &= \text{Pr}(x_{0:L})_{\pi} \\
    &= \text{tr}(\mathcal{E}^{x_L}\circ \mathcal{E}^{x_{L-1}}\circ \mathcal{E}^{x_0}[\rho_m])~\forall~x_{0:L}, 
\end{align}
where $\rho_m$ corresponds to the quantum initial memory which gives rise to the same statistics as those of the HMM initialized in the state $\pi$. We introduce the following (generally non-orthogonal) states
\begin{align}\label{eq:memory_states}
|\psi_i \rangle := \sum_{j,x} e^{\iota\lambda(i,j,x)}\sqrt{T^{(x)}_{i \rightarrow j}}|{j}\rangle\!\!\ket{x}.
\end{align}
where $\{\ket{j}\}$ and $\{\ket{x}\}$ are orthonormal bases corresponding to the HMM's hidden states and to the alphabet, respectively.  $e^{\iota\lambda(i,j,x)}$ is a phase factor that does not affect measurement statistics. We define $\mathcal{H}_M:=\text{span}\{\ket{\psi_m}\}$ as the memory Hilbert space. 
Then, in order to generate the same process as Eq.~\eqref{eq:prbabilties in terms of instruments}, we enforce that the initial state encodes the given classical distribution $\pi$ as,
\begin{align}
    \rho_m  =\sum_{i} \pi_i|\psi_{i}\rangle\!\langle{\psi_{i}}|.
\end{align}
Note here that generally $\langle{\psi_i}|\psi_j\rangle\neq\delta_{ij}$. 
The CP maps $\{\mathcal{E}^{(x)}\}$ which generate the process are
\begin{align}\label{eq: update map}
    \mathcal{E}^{(x)}[\rho] &= \sum_{k}K^{(x)}_{k} \rho K^{(x)\dag}_{k}
\end{align}
where $\{K^{(x)}_{k}\}$ are the Kraus operators defined as
\begin{align}\label{eq: Kraus for Example}
    K^{(x)}_{k}:= |{\psi_k}\rangle\!\langle k|\!\langle x|
\end{align}  
and the maps $\{\mathcal{E}^{x}\}$ correspond to those required for Eq.~\eqref{eq:prbabilties in terms of instruments}. The map $\sum_x\mathcal{E}^{(x)}$ is CPTP. If $\rho=\rho_m$, then Eq.~\eqref{eq: update map} becomes
\begin{align}
    \mathcal{E}^{(x)}[\rho_m] &= \sum_{i,j}T^{(x)}_{i \rightarrow j}\pi_i\ket{\psi_j}\!\!\bra{\psi_j}\text{, where}\\
    \text{Pr}(x)_{\pi}&=\tr[\mathcal{E}^{(x)}[\rho_m]]=\overrightarrow{1}^TT^{(x)}\pi,
\end{align}
and likewise for repeated application of the instrument $\mathcal E$ and correspondingly longer strings $x_{0:L}$. {The construction given in Eqs.~\eqref{eq:memory_states}, \eqref{eq: update map} and \eqref{eq: Kraus for Example} closely follows the approach of Ref.\cite{elliott_memory_2021}, which provides a similar mapping from classical HMMs to QHMMs.}
\\
\noindent Lastly, we show that instrument $\{\mathcal{E}^{(x)}\}$ leaves $H_M$ invariant by applying an arbitrary Kraus operator $K_k^{(y)}$ to an arbitrary~$\ket{\psi_i}$, noting that $\{\ket{\psi_i}\}$ forms a basis of $\mathcal{H}_M$ by definition (non-orthogonal, and potentially overcomplete):
\begin{align}
    K_k^{(y)}\ket{\psi_i}=e^{\iota\lambda(i,k,y)}\sqrt{T^{(y)}_{i \rightarrow k}}\ket{\psi_k}\in \mathcal{H}_M
\end{align}
This means that Kraus operators $K_k^{(y)}$ of this QHMM only interact with the memory Hilbert space $\mathcal{H}_M$.   Thus, we define a new set of Kraus operations that act on the memory Hilbert space alone:
\begin{align}\label{eq:Kraus for qhmm}
\tilde K_k^{(y)}:=P_M K_k^{(y)} P_M,
\end{align}
where $P_M$ is a projector corresponding to $\mathcal{H}_M$. This new set of Kraus operators $\{\tilde K_j^{(x)}\}$ fulfills a completeness relation with respect to the state space corresponding to $\mathcal{H}_M$. That is, for any state $\rho_M$ acting on $\mathcal{H}_M$ the instrument $\tilde{\mathcal{E}}:=\{\tilde{\mathcal{E}}^{(x)}\}$ resulting from these Kraus operators is CPTP on the restricted space.  Moreover, the associate word probabilities produced by this instrument are identical to the original uncompressed map
\begin{align}
\text{tr}(\tilde{\mathcal{E}}^{x_L}\circ \tilde{\mathcal{E}}^{x_{L-1}}\circ \tilde{\mathcal{E}}^{x_0}[\rho_m])=\text{tr}(\mathcal{E}^{x_L}\circ \mathcal{E}^{x_{L-1}}\circ \mathcal{E}^{x_0}[\rho_m]).
\end{align}
Therefore, whenever the set of states $\{\ket{\psi_m}\}$ are linearly dependent, the QHMM thus constructed requires a smaller dimensional Hilbert space than the number of hidden states of the HMM from which it was constructed.\\

\noindent The existence of memory-reduced models is demonstrated by means of a minimalist example in the main text of this paper.
\end{onecolumngrid}

\end{document}